\documentclass[aps,onecolumn,twoside,floatfix,pra,nofootinbib]{revtex4-2}
\usepackage[utf8]{inputenc}
\usepackage[T1]{fontenc}
\usepackage{graphicx}
\usepackage{amsmath}
\usepackage{amsthm}
\usepackage{amssymb}
\usepackage[a4paper,margin=2.3cm]{geometry}
\usepackage{hyperref}
\usepackage{braket}
\usepackage{bbold}
\usepackage{mathtools}
\usepackage{xparse}
\usepackage{xcolor}
\def\urlprefix{}
\hypersetup{
    colorlinks,
    linkcolor={red!50!black},
    citecolor={blue!50!black},
    urlcolor={blue!80!black}}

\ExplSyntaxOn
\NewDocumentCommand{\Rowvec}{ O{,} m }
 {
  \vector_main:nnnn { p } { & } { #1 } { #2 }
 }
\NewDocumentCommand{\Colvec}{ O{,} m }
 {
  \vector_main:nnnn { p } { \\ } { #1 } { #2 }
 }

\seq_new:N \l__vector_arg_seq
\cs_new_protected:Npn \vector_main:nnnn #1 #2 #3 #4
 {
  \seq_set_split:Nnn \l__vector_arg_seq { #3 } { #4 }
  \begin{#1matrix}
  \seq_use:Nnnn \l__vector_arg_seq { #2 } { #2 } { #2 }
  \end{#1matrix}
 }
\ExplSyntaxOff

\DeclareMathAlphabet{\mathpzc}{OT1}{pzc}{m}{it}

\newcommand\nctwo{\stackrel{\mathclap{\normalfont\mbox{\tiny (1)}}}{\leq}}
\newcommand\ncthree{\stackrel{\mathclap{\normalfont\mbox{\tiny (2)}}}{=}}
\newcommand{\nbox}[2][9]{\hspace{#1pt} \mbox{#2} \hspace{#1pt}}
\newcommand{\pave}{p_{\textnormal{ave}}}
\newtheorem{lem}{Lemma}
\newtheorem{cor}{Corollary}
\DeclareMathOperator{\tr}{tr}

\DeclareMathOperator{\spec}{spec}
\DeclareMathOperator{\rank}{rank}

\begin{document}
\title{Quantifying incompatibility of quantum measurements through non-commutativity}
\author{Krzysztof Mordasewicz}
\email{k.mordasewicz@student.uw.edu.pl}
\author{J\k{e}drzej Kaniewski}
\email{jkaniewski@fuw.edu.pl}
\affiliation{Faculty of Physics, University of Warsaw, Pasteura 5, 02-093 Warsaw, Poland}
\date{\today}
\begin{abstract}
The existence of incompatible measurements, i.e.~measurements which cannot be performed simultaneously on a single copy of a quantum state, constitutes an important distinction between quantum mechanics and classical theories. While incompatibility might at first glance seem like an obstacle, it turns to be a necessary ingredient to achieve the so-called quantum advantage in various operational tasks like random access codes or key distribution. To improve our understanding of how to quantify incompatibility of quantum measurements, we define and explore a family of incompatibility measures based on non-commutativity. We investigate some basic properties of these measures, we show that they satisfy some natural information-processing requirements and we fully characterize the pairs which achieve the highest incompatibility (in a fixed dimension). We also consider the behavior of our measures under different types of compositions. Finally, to link our new measures to existing results, we relate them to a robustness-based incompatibility measure and two operational scenarios: random access codes and entropic uncertainty relations.
\end{abstract}

\maketitle
\tableofcontents

\section{Introduction}
Measurement incompatibility constitutes an important distinction between quantum mechanics and classical physics.
The earliest mention of incompatibility might perhaps be identified with the Heisenberg uncertainty principle, which states that it is impossible to simultaneously measure the position and momentum of a quantum particle with arbitrary precision~\cite{HW1927}. A further development by Robertson clearly shows that this uncertainty arises as a consequence of two observables failing to commute~\cite{Commute_first}, which constitutes a direct link to the concept of incompatibility.

While incompatibility might seem like an obstacle, in many information-processing tasks it turns out to be a prerequisite in order to achieve a better-than-classical performance. Among tasks of predominantly foundational relevance this holds for Einstein--Podolski--Rosen steering \cite{EPR_first, EPR_Bell_nonlocality} or Bell nonlocality \cite{Bell_nonlocality_first, EPR_Bell_nonlocality, Nonlocality_review}.
The same is true for communication problems. In a quantum random access code (QRAC), where we want to encode two classical dits into a single qudit, it is necessary to use incompatible quantum measurements in order to beat the performance of the optimal classical strategy. Similarly, security of a quantum key distribution (QKD) protocol cannot be guaranteed unless incompatible measurements are used.

It is intuitively clear that in all the scenarios mentioned above in order to achieve high performance (e.g.~large Bell violation, high QRAC performance, high security against the eavesdropper) one should be using highly incompatible measurements. For certain special classes of measurements, e.g.~pairs of rank\nobreakdash-1 projective measurements on a qubit, one can in fact formalize this intuition through rigorous quantitative relations. The obstacle we face when attempting to generalize such results to more complex scenarios is the lack of a natural and unique method to quantify how incompatible a pair of measurements is. This motivates the investigation of incompatibility measures. Let us now describe some qualitative notions of compatibility found in literature and quantitative measures based on them. Note that the list below is not meant to be exhaustive. For more details we refer the reader to works on layers of classicality, e.g.~Refs.~\cite{PhysRevA.93.042118, PhysRevA.104.022206}.

The least restrictive (in the sense of having the largest set of compatible measurements) notion of compatibility found in literature is \textbf{coexistence} \cite{Coexistance_first}. A pair of measurements is coexistent if and only if there exists a single parent measurement such that every measurement operator from the original pair can be obtained by adding up a subset of measurement operators of the parent measurement. Coexistence does not seem to be a popular research topic and we are not aware of any attempts to use this notion to construct a quantitative measure.

A slightly more restrictive notion is that of \textbf{joint measurability}. We say that two measurements are jointly measurable, if there exists a single parent measurement from which both original measurements can be recovered through classical postprocessing.
The set of jointly measurable pairs of measurements acting on $\mathbb{C}^{d}$ is a convex set and therefore some of the resulting measures can be computed efficiently.

Incompatibility measures based on joint measurability can be constructed using standard tools of convex geometry. Suppose we are given two convex sets $A, B$ such that $A \subsetneq B$ and our goal is to quantify the distance between a point $x \in B \setminus A$ and the set $A$.\footnote{Analogous approach is often used in the context of resource theories where $A$ is the set of free (resource-less) objects and the distance of $x$ from $A$ quantifies how resourceful $x$ is.} Convex geometry provides two standard approaches to this problem which we refer to as the weight approach and the robustness approach. In the weight approach we search for a convex decomposition of $x$ into a single point from $A$ and a single point from $B$ which assigns the largest weight to the point from $A$. In the robustness approach we quantify how difficult it is to ``push'' $x$ inside set $A$ by taking a convex combination with some point in $B$. 
These two approaches have been successfully used to quantify various quantum resources, e.g.~entanglement~\cite{Entanglement_first, Entanglement_review}, steering~\cite{Steering_first_robustness,Steering_first_weight, Steering_review} or nonlocality~\cite{Nonlocality_first,Nonlocality_review}. Since these problems can often be cast as convex optimisation problems, we can study them analytically. Moreover, if the convex sets $A$ and $B$ admit compact descriptions, it might be possible to efficiently compute the resulting quantities.

Applying this standard approach to incompatibility gives rise to two families of incompatibility measures: incompatibility weight~\cite{IW} and incompatibility robustness~\cite{IR}.
Unfortunately, even though these measures are elegant from the mathematical point of view, they do not fully capture our physical intuition. The main drawback of the incompatibility weight is its poor performance for certain important classes of measurements. More specifically, if we consider a pair of rank-1 measurements, then they are either compatible or maximally incompatible.\footnote{Since rank-1 measurements cannot be decomposed in a non-trivial manner, in this case the incompatibility weight only takes the extreme values, i.e.~$0$ or $1$.}
Incompatibility robustness, on the other hand, requires us to define the noise set (denoted by $B$ in the explanation above). In the modern literature on incompatibility there are five distinct choices of the noise set and not only do these choices lead to different measures, but they even lead to a different ordering of measurements~\cite{DFK19}. In fact, for just one of the measures we can give an example of the most incompatible measurements in arbitrary dimension, but even then we do not have a full characterization of the most incompatible measurements. For the other measures we do not know any examples of the most incompatible measurements and for most of them we are not even aware of any reasonable conjectures.

Another notion of compatibility found in literature is that of \textbf{non-disturbance}~\cite{ND}. This notion has a well-defined operational meaning: we say that measurements $E$ and $F$ are non-disturbing if it is possible to perform $E$ without affecting the statistics of the subsequent measurement $F$ on a single copy of the system. Clearly, this definition implies that $E$ and $F$ are not treated on equal footing and indeed the notion of non-disturbance, unlike all the other notions considered in this work, is not symmetric \cite{ND}. This asymmetry is a natural consequence of the definition, but it shows that non-disturbance should not be compared with the other notions of incompatibility.\footnote{Nevertheless, for the special case of projective measurements on a qubit there exists a quantitative measure of non-disturbance, which turns out to be closely related to the incompatibility measures defined in this work (see Section V.A of Ref.~\cite{ND}).}

The most restrictive, and arguably the simplest notion of compatibility corresponds to \textbf{commutativity} of the measurement operators. The set of commuting measurements is not jointly convex,
but it is easy to check, whether a pair of measurements commutes.\footnote{It is worth pointing out that an interesting connection between joint measurability of measurements and commutativity of their Naimark dilations has been proven for pairs~\cite{BENEDUCI2017197} and larger sets of measurements~\cite{mitra2021characterizing}.} Even though, as mentioned at the beginning, the phenomenon of non-commutativity goes back to the early days of quantum mechanics, we are not aware of any attempts at quantifying it, which is precisely the main focus of this work.

In this work we propose a family of measures based on non-commutativity and we study their properties. We show that they exhibit some natural properties that an incompatibility measure should satisfy, e.g.~they are invariant under unitaries and they do not increase under post-processing.
We compute the largest value of these measures in a fixed dimension and we fully characterize the pairs of measurements that achieve it, which leads to a natural generalisation of mutually unbiased bases.
We also investigate the behavior of the measures under different types of compositions.
Finally, we show that any measurements which are maximally incompatible according to our measures must be maximally incompatible according to one of the robustness-based incompatibility measures, but we do not know whether the reverse statement holds. For the special case of rank-1 projective measurements we derive a quantitative relation between commutation-based incompatibility and two operational scenarios: performance in a QRAC and entropic uncertainty relations.

In Section~\ref{sec:preliminaries} we define some standard mathematical concepts used throughout the paper.
In Section~\ref{sec:inc_measure_def_prop} we define a new family of incompatibility measures, study its properties and fully characterize the most incompatible measurements.
In Section~\ref{sec:comparison_exist_results} we compare the new family with previously used incompatibility measures and we investigate its usefulness for operational tasks.

\section{Preliminaries}
\label{sec:preliminaries}
\subsection{Vectors and operators}
For a vector $v \in \mathbb{C}^d$  for $p\in [1,\infty]$ we define the vector $p$-norm as:
\begin{align*}
    \| v \|_p := 
    \begin{cases}
      \left( \sum_{j} |v_j|^p \right)^\frac{1}{p}, \; &\text{for} \; p \in [1,\infty[,\\
      \max_j |v_j|, \; &\text{for} \; p = \infty,
    \end{cases}
\end{align*}
where $v_{j}$ denotes the $j$-th component of $v$, while the summation and the maximum goes over $j \in \{1, 2, \ldots, d\}$.

Let $A : \mathbb{C}^d \to \mathbb{C}^d$ be a linear operator and let $A^{\dagger}$ be its hermitian conjugate. We denote the identity operator acting on $\mathbb{C}^d$ by $\mathbb{1}_{d}$ although we omit the subscript where the dimension is clear from the context. A complex number $\lambda$ is an eigenvalue of $A$ if $\det(A- \lambda \mathbb{1}) = 0$, where $\det(\cdot)$ denotes the determinant.
The set of eigenvalues is called the spectrum of $A$ and is denoted by $\spec(A)$. 
We say that $A$ is normal if it satisfies $[A,A^\dagger]=0$ and then the rank of $A$, denoted by $\rank(A)$, equals the number of non-zero eigenvalues (including multiplicities).
Every normal operator $A$ can be written in its spectral decomposition: $A = \sum_j \lambda_j P_j$, where $\lambda_j$ is an eigenvalue of $A$ and $P_j$ is the corresponding eigenprojector.

If all eigenvalues of a normal operator $A$ are real, i.e.~$\lambda \in \mathbb{R}$, which is equivalent to $A^\dagger = A$, we say that $A$ is hermitian. If all eigenvalues are also non-negative, i.e.~$\lambda \geq 0$, we say that $A$ is positive semi-definite, written as $A \geq 0$. If all eigenvalues of a normal operator are purely imaginary (equivalent to $A^\dagger = -A$), we say that $A$ is anti-hermitian.

For a hermitian operator $A$ we define its positive and negative parts as follows: $A^+ := \sum_{j : \lambda_{j} > 0} \lambda_{i} P_{j}$ and $A^- := \sum_{j : \lambda_{j} < 0} \lambda_{j} P_{j}$. For a normal operator we define the absolute value of $A$ as $|A| := \sum_i |\lambda_i| P_j$ and clearly $|A| \geq 0$. It is easy to see that for hermitian operators we have $|A| = A^+ - A^-$.

Every linear operator $A : \mathbb{C}^d \to \mathbb{C}^d$ can be written in its singular value decomposition. The singular value decomposition uniquely determines the singular values of $A$, which we denote by $\{ \sigma_{j} \}$. The Schatten $p$-norm of $A$ is defined as
\begin{align*}
    \| A \|_p :=
    \begin{cases}
      \big( \sum_j \sigma_j^p \big)^\frac{1}{p}, \; &\text{for} \; p \in [1,\infty[,\\
      \max_j \, \sigma_j, \; &\text{for} \; p=\infty,
    \end{cases}
\end{align*}
where the summation and the maximum goes over all the singular values of $A$ (including multiplicities). Note that the Schatten $p$-norm of $A$ is simply the vector $p$-norm of the vector of singular values of $A$.

We say that two norms, $\|\cdot\|_{p}$ and $\|\cdot\|_{q}$, are dual to each other if either $p = 1, q = \infty$ or $p = \infty, q = 1$ or $\frac{1}{p} + \frac{1}{q} = 1$. This allows us to write the Schatten $p$-norm as an optimization problem:
\begin{align}
\label{eq:schatten-opt}
    \| A \|_p := \max \left\{ | \braket{X,A} | \, : \, \| X \|_q \leq 1 \right\},
\end{align}
where $\|\cdot\|_{p}$ and $\|\cdot\|_{q}$ are dual norms and $\braket{X,Y} := \tr( X^\dagger Y)$ is the Hilbert--Schmidt inner product (see Lecture 2 of Ref.~\cite{TQI_notes}). In fact, it is easy to see that the optimal operator $X$ must satisfy $\| X \|_q = 1$.

\subsection{Convex sets and concave functions}
\label{sec:concave_functions}
Let $S$ be a convex subset of $\mathbb{R}^{n}$. A point $x \in S$ is called an extremal point of $S$ if it cannot be written as a non-trivial convex combination of some other points from $S$. Let us now present a complete characterisation of the extremal points of a symmetric polytope related to the probability simplex.

\begin{lem}
\label{lem:extremal_points}
Let $n \geq 2$ be an integer and $s, t$ be positive real numbers satisfying $t n \geq s$. Let $S$ be a subset of $\mathbb{R}^{n}$ defined as
\begin{align*}
S := \Big\{ x \in \mathbb{R}^{n} : 0 \leq x_{j} \leq t \; \forall j \nbox{and} \sum_{j = 1}^{n} x_{j} = s \Big\}.
\end{align*}
Then, $S$ forms a non-empty convex set and, moreover, its extremal points are permutations of the vector $u$ given by
\begin{align*}
u_{j} =
\begin{cases}
t &\nbox{for} j \leq \lfloor \frac{s}{t} \rfloor,\\
s - \lfloor \frac{s}{t} \rfloor \cdot t &\nbox{for} j = \lfloor \frac{s}{t} \rfloor + 1,\\
0 &\nbox{otherwise.}
\end{cases}
\end{align*}
\end{lem}
\noindent As this lemma is rather straightforward, let us just sketch the argument. It is easy to see that points which have two (or more) components in the interior $(0, t)$ cannot be extremal (introducing small variations to these components leads to a non-trivial convex decomposition). Once we know that there is at most one component taking some intermediate value, we are left with precisely the points presented above. This shows that the extremal points of $S$ are contained in the permutations of $u$. To show that the sets are actually the same it suffices to prove that no permutation of $u$ can be written as a convex combination of all the other permutations of $u$, which is a simple exercise.

In the second part of this section, let us state a couple of well-known facts about concave functions. Let $S$ be a convex subset of $\mathbb{R}^{n}$. We say that a function $f : S \to \mathbb{R}$ is concave if for all $x, y \in S$ and all $p \in [0, 1]$ we have
\begin{align*}
f( p x + (1 - p) y ) \geq p f(x) + (1 - p) f(y).
\end{align*}
We say that a function is strictly concave if the equality holds only in the trivial cases, i.e.~$p \in \{0, 1\}$ or $x = y$. Concave functions are convenient to work with because they satisfy Jensen's inequality, i.e.~for an arbitrary choice of $x_{j} \in S$ and an arbitrary choice of non-negative weights $p_{j}$ that add up to unity we have
\begin{align}
\label{eq:Jensen}
\sum_{j} p_{j} f (x_{j} ) \leq f \Big( \sum_{j} p_{j} x_{j} \Big).
\end{align}
This means that if we want to obtain an upper bound on the left-hand side of Eq.~\eqref{eq:Jensen}, it suffices to compute the expectation value over the arguments. On the other hand, if our goal is to minimize $f(x)$ over some convex and compact set $S$, then it is sufficient to perform the minimisation over the extremal points of $S$ (every non-extremal point can be decomposed into extremal points and at least one of them will give a value which is at least as good as the original non-extremal point).

Finally, let us mention a connection between concave functions and majorization. Let $y, z \in \mathbb{R}^{n}$ be a pair of real vectors with non-negative coefficients. We say that $y$ majorizes $z$, denoted by $y \succ z$, if for all $k \in \{1, 2, \ldots, n\}$ we have:
\begin{align}
\label{eq:majorization-condition}
\sum_{j = 1}^{k} y^{\downarrow}_{j} \geq \sum_{j = 1}^{k} z^{\downarrow}_{j},
\end{align}
where the vector $y^{\downarrow}$ ($z^{\downarrow}$) has the same components as $y$ ($z$) but sorted in descending order. Let $f(x) : \mathbb{R} \to \mathbb{R}$ be a concave function and let $g : \mathbb{R}^{n} \to \mathbb{R}$ be defined as
\begin{align*}
g(y) := \sum_{j = 1}^{n} f(y_{j}).
\end{align*}
Then, $g$ is a Schur-concave function, i.e.~for a pair of vectors satisfying $y \succ z$, we have $g(y) \leq g(z)$.
\subsection{Quantum measurements}
In this section we introduce the concept of a quantum measurement and explain some transformations that might be applied to it.

In quantum mechanics the most general class of measurements is known as positive-operator valued measures (POVMs). 
A measurement with $n$ outcomes acting on a $d$-dimensional space is defined by a set of $n$ positive semi-definite operators acting on $\mathbb{C}^{d}$, which we denote by $ \{ {E}_a \}_{a=1}^{n}$, satisfying $\sum_{a=1}^n E_a = \mathbb{1}$ and we will refer to this measurement as $E$.
According to the Born rule if we perform this measurement on a state $\rho$, the probability of observing outcome $a$ is given by $\tr( E_{a} \rho )$.
In this work we assume that every outcome corresponds to a non-zero measurement operator, i.e.~$E_a \neq 0$ for all $a$.

Let us define two special classes of measurements:
\begin{itemize}
    \item A measurement is \textbf{projective} if every measurement operator is a projector: $E_{a}^{2} = E_{a}$. Then, the number of outcomes cannot be larger than the dimension: $n \leq d$.
    \item A measurement is \textbf{rank-1} if every measurement operator is rank-1: $\rank(E_a) = 1$. Then, the number of outcomes cannot be smaller than the dimension: $n \geq d$.
\end{itemize}

A measurement which is simultaneously projective and rank-1 is sometimes referred to as a \textbf{measurement in a basis}, because its measurement operators must be of the form $E_{a} = \ket{a}\bra{a} $, where $\{ \ket{a} \}_{a = 1}^{d}$ forms an orthonormal basis of $\mathbb{C}^{d}$. For such measurements the number of outcomes equals the dimension of the space.

Let $\{ \ket{e_j} \}$ and $\{ \ket{f_k} \}$ be two orthonormal bases on $\mathbb{C}^d$. We say that these bases are mutually unbiased if
\begin{align*}
    |\braket{e_j|f_k}|^2 = \frac{1}{d}, \; \; \; \forall j,k \in \{1, \dots, d\}.
\end{align*}
The standard example of mutually unbiased bases (MUBs) for qubits is given by $\{ \ket{0}, \ket{1} \}$ and $\{ \frac{1}{\sqrt{2}}(\ket{0} + \ket{1}), \frac{1}{\sqrt{2}}(\ket{0} - \ket{1})\}$.

Having discussed what a measurement is, let us now describe two ways in which a measurement can be transformed.

\textbf{Post-processing} refers to classically processing the outcomes of a quantum measurement. If $P_{a' | a}$ is the probability that the initial outcome $a$ results in the final outcome $a'$, then the post-processing is described by a set of real numbers satisfying  $P_{a'|a} \geq 0$ and $\sum_{a'} P_{a'|a} = 1$. Clearly, the obtained statistics are equivalent to those obtained by performing a measurement $\widetilde{E}$, where
\begin{align}
\label{eq:def_post-processing}
  \widetilde{E}_{a'} = \sum_{a} P_{a'|a} E_{a}. 
\end{align}
Post-processing might change the number of outcomes, but it preserves the dimension on which the measurement acts.

\textbf{Pre-processing} refers to putting a quantum state through a quantum channel (i.e.~a completely positive trace-preserving map) before feeding it into the measurement device.
Let $\rho$ be a density matrix acting on dimension $d'$ and let $\Lambda : L(\mathbb{C}^{d'}) \xrightarrow{} L(\mathbb{C}^{d})$ be a quantum channel given by
\begin{align*}
    \Lambda(\rho) = \sum_j K_j \rho K_j^\dagger,
\end{align*}
where $\{ K_j \}$ are $d \times d'$ operators satisfying
\begin{align*}
    \sum_j K_j^\dagger K_j = \mathbb{1}_{d'}.
\end{align*}
Performing a measurement $E$ on $\Lambda(\rho)$ is equivalent to performing a measurement $F$ given by
\begin{align*}
F_{a} := \Lambda^\dagger(E_{a}) \text{ for} \; \Lambda^\dagger (A) = \sum_j K_j^\dagger A K_j,
\end{align*}
on the original state $\rho$. We say that the measurement $F$ is a pre-processing of the original measurement $E$. Note that pre-processing might change the dimension on which the measurement acts but it preserves the number of outcomes.

\subsection{Technical lemmas}
In this section we present some technical lemmas that we will use throughout this work.
Let us start with the triangle inequality for Schatten norms. For an arbitrary pair of operators $A$ and $B$ acting on $\mathbb{C}^d$ and any Schatten $p$-norm it holds that
\begin{align*}
    \| A + B \|_p \leq \|A\|_p + \|B\|_p.
\end{align*}
In the following lemma we completely characterize the cases in which the triangle inequality holds as equality for hermitian operators $A$ and $B$.
\begin{lem}
\label{lem:tri_ineq_sat}
For hermitian operators $A$ and $B$ the triangle inequality holds as equality, i.e.
\begin{align}
\label{eq:opt_triangle_ineq}
    \| A + B \|_p = \|A\|_p + \|B\|_p,
\end{align}
if and only if there exists a hermitian operator $X$ satisfying $\| X \|_q \leq 1$,
which is optimal for both $A$ and $B$, i.e.
\begin{align}
    \label{eq:norm-saturation}
    \braket{X, A} = \| A \|_p \, \textrm{ and } \, \braket{X,B} = \| B \|_p.
\end{align}
\end{lem}
\begin{proof}
The non-trivial direction is to show that Eq.~\eqref{eq:opt_triangle_ineq} implies the existence of an operator $X$ satisfying the conditions stated above.
Let us start by showing that for every hermitian operator $H$ we can always find another hermitian operator $X$ such that $\braket{X,H} = \| H \|_p$ and $\| X \|_q \leq 1$.
The definition given in Eq.~\eqref{eq:schatten-opt} implies that there always exists some matrix $Y$ satisfying $\|Y\|_{q} = 1$ such that $|\braket{Y,H}| = \| H \|_p$.
If we now omit the absolute value, we conclude that $\braket{Y,H} = \| H \|_p \cdot e^{i \phi}$ for some phase $\phi \in [0, 2\pi[$. Then we construct a new matrix $Z := Y \cdot e^{i \phi}$, so that $\braket{Z,H} = \| H \|_p$.
Finally, we construct $X$ by symmetrizing $Z$: $X := \frac{Z + Z^\dagger}{2}$.
Since $\|Z\|_q \leq 1$ and $\|Z^\dagger \|_q \leq 1$, we immediately conclude that $\|X\|_q \leq 1$. Since $H$ and $X$ are hermitian, we deduce that
\begin{align*}
    \braket{X,H} = \frac{1}{2} \big(
    \braket{Z,H} +
    \braket{Z^\dagger,H}
    \big) = \braket{Z,H} = \| H \|_p.
\end{align*}
If we now set $H = A+B$, we conclude that there exists a hermitian operator $X$ such that $\| A+B \|_p = \braket{X,A+B} = \braket{X,A} + \braket{X,B}$. Each term on the right-hand side can be upper bounded by the norm, so if the triangle inequality holds as equality, then $X$ must satisfy the conditions given in Eq.~\eqref{eq:norm-saturation}.
\end{proof}
\begin{lem}
\label{lem1}
Let $\|\cdot\|_{p}$ and $\|\cdot\|_{q}$ be dual norms. Let $A$ and $X$ be hermitian operators, such that $\| X \|_q = 1$ and $\langle A,X \rangle = \| A \|_p$. Then the positive (negative) part of $A$ is orthogonal to the negative (positive) part of $X$, i.e.:
\begin{align*}
  A^{+} X^{-} = A^{-} X^{+} = 0.
\end{align*}
\end{lem}
\begin{proof}
Let us write $A$ and $X$ in their spectral  decompositions:
\begin{align*}
    A &= \sum_j \lambda_j \, P_j,\\
    X &= \sum_k \pi_k \, R_k.
\end{align*}
It is easy to check that
\begin{align*}
    \langle A, X \rangle = \sum_{j k} \lambda_j \pi_k \tr (P_j \, R_k) \leq \sum_{j k} |\lambda_j| |\pi_k| \tr (P_j \, R_k) = \langle |A|, |X| \rangle \leq \bigl\| \; |A| \; \bigr\|_p = \| A \|_p,
\end{align*}
where the second inequality comes from the fact that $\bigl\| \, | X | \, \bigr\|_q = \| X \|_q = 1$. However, by assumption the left-hand side equals $\| A \|_p$. This implies that the first inequality has to be tight term-by-term, which implies that whenever $\lambda_j \pi_k < 0$, we must have $\tr (P_j \, R_k) = 0$. Since $P_{j}$ and $R_{k}$ are projectors, $\tr (P_j \, R_k) = 0$ implies $P_j \, R_k = 0$, which gives precisely the orthogonality condition stated in the lemma.
\end{proof}

\begin{lem}
\label{lem:spec_commutator}
Let $A$ and $B$ be positive semi-definite rank-1 operators, i.e.
\begin{align*}
    A &= \alpha \ket{e} \bra{e},\\
    B &= \beta \ket{f} \bra{f}
\end{align*}
for some non-negative numbers $\alpha, \beta \geq 0$ and vectors $\ket{e}, \ket{f}$. Then,
\begin{enumerate}
    \item[(a)] the non-zero eigenvalues of $[A,B]$ are given by $\{ \pm i \lambda \}$,
    \item[(b)] the non-zero eigenvalues of $A+B$ are given by $\{ \eta_\pm\}$,
\end{enumerate}
where
\begin{align*}
    \lambda &:= \alpha \beta \, |\braket{e|f}| \, \sqrt{1-|\braket{e|f}|^2},\\
    \eta_\pm &:= \frac{1}{2} ( \alpha + \beta \pm \sqrt{(\alpha-\beta)^2 + 4 \alpha \beta |\braket{e|f}|^2} ).
\end{align*}
In both cases the non-zero eigenvalues are non-degenerate.
\end{lem}
\begin{proof}
Let us first prove part (a). The commutator of $A$ and $B$ is an anti-hermitian operator and its rank is at most 2: $\rank ([A,B]) \leq 2$.
Moreover, since $\tr ([A,B]) = 0$ the non-zero part of the spectrum must be of the form $\{ i\lambda,-i\lambda \}$ for some $\lambda \geq 0$. Let us now calculate the square of the commutator:
\begin{align*}
    [A,B]^2 = ABAB + BABA - AB^2 A - BA^2 B.
\end{align*}
From linearity and cyclicity of the trace we have:
\begin{align*}
    \tr (ABAB) &= \tr (BABA) = \alpha^2 \beta^2 |\braket{e|f} |^4,\\
    \tr (AB^2 A) &= \tr (BA^2 B) = \alpha^2 \beta^2 |\braket{e|f} |^2.
\end{align*}
Combining this with the observation that $\tr ([A,B]^2) = -2 \lambda^2$ leads to the final result of part (a).

For part (b) we use the same method. The sum of A and B is at most rank-2: $\rank(A+B) \leq 2$, hence, the non-zero part of the spectrum is of the form $\{ \eta_+, \eta_- \}$. Computing the square gives
\begin{align*}
    (A+B)^2 = A^2 + AB + BA + B^2.
\end{align*}
Once again we use linearity and cyclicity of the trace to get:
\begin{align*}
    \tr(A+B) &= \alpha + \beta\\
    \tr(A^2) &= \alpha^2\\
    \tr(B^2) &= \beta^2\\
    \tr(AB) &= \tr(BA) = \alpha \beta |\braket{e|f}|^2.
\end{align*}
Finally solving a simple quadratic equation: $\tr(A+B)= \eta_+ + \eta_-$ and $\tr((A+B)^2) = \eta_+^2 + \eta_-^2$ leads to the final result of part (b).
\end{proof}
Let us now introduce the convention that $2^\frac{1}{p} = 1$ when $p = \infty$, which will be used throughout this paper.
\begin{cor}
\label{Corollary:norm_of_com}
For operators $A$ and $B$ as defined in Lemma~\ref{lem:spec_commutator} it holds that
\begin{align*}
    \| [A,B] \|_p = \left( 2 |\lambda|^p \right)^\frac{1}{p} = 2^\frac{1}{p} \cdot \alpha \beta \, c \, \sqrt{1-c^2},
\end{align*}
where $c := |\braket{e|f}|$ is the overlap, i.e.~the modulus of the scalar product between the eigenvectors of $A$ and $B$.
\end{cor}
To simplify the notation let us define
\begin{align}
\label{eq:h_p_def}
   h_p (c) := 2^\frac{1}{p} \cdot c \, \sqrt{1-c^2},
\end{align}
which allows us to write the $p$-norm of the commutator as
\begin{align*}
    \| [A,B] \|_p = \alpha \beta h_p (c).
\end{align*}
\section{A new family of incompatibility measures and its basic properties}
\label{sec:inc_measure_def_prop}
We are now ready to explicitly state what an incompatibility measure is and what properties we would like it to satisfy. Let $\mathcal{M}_{d, n}$ be the set of measurements with $n$ outcomes acting on $\mathbb{C}^{d}$. For our purposes an incompatibility measure is a function $\Upsilon : \mathcal{M}_{d, n_{E}} \times \mathcal{M}_{d, n_{F}} \to \mathbb{R}$. Note that the number of outcomes of the two measurements may differ, but the dimension they act on must be the same. Let us now list some natural properties that a reasonable measure should satisfy.
\begin{itemize}
    \item The measure should be non-negative for all pairs of measurements.
    \item The measure should be equal to zero for a natural and easy to characterize class of measurements.
    \item The measure should be non-increasing in some natural scenarios in which it should not be possible to generate incompatibility.
    \item We should be able to determine the maximal value of the measure and fully characterize the pairs of measurement that achieve it.
\end{itemize}

\subsection{Definition and basic properties}
In the remainder of this section we define a new family of incompatibility measures, we study their basic properties and fully characterize the most incompatible measurements.
For $p \in [1, \infty]$ the incompatibility of
measurements $\{ E_a \}_{a=1}^{n_E}$ and $\{ F_b \}_{b=1}^{n_F}$ is given by
\begin{align*}
    \Upsilon_p (E,F) := \sum_{a b} \|[E_a,F_b]\|_p,
\end{align*}
where the summation is taken over $a \in \{1, 2, \ldots, n_{E}\}$ and $b \in \{1, 2, \ldots, n_{F}\}$.

Clearly, these incompatibility measures are non-negative and since Schatten norms vanish if and only if the operator vanishes, the measures vanish if any only if all the measurement operators commute. 
This addresses the first two requirements mentioned in the previous section.

The measures are not affected if a unitary operation is applied to both measurements. 
It is sufficient to note that for an arbitrary unitary $U$ we have
\begin{align*}
    [U E_a U^\dagger, U F_b U^\dagger] = U E_a U^\dagger U F_b U^\dagger - U F_b U^\dagger U E_a U^\dagger = U E_a F_b U^\dagger - U F_b E_a U^\dagger = U [E_a,F_b] U^\dagger
\end{align*}
and use the fact that Schatten norms are invariant under unitaries.

\phantomsection
\label{sec:post-processing}
Since post-processing is a classical operation we expect that it should not increase incompatibility, which is an inherently quantum property. 
Let us now show that our measures do not increase under post-processing. 
Consider a pair of measurements $E$ and $F$ and suppose that $E$ undergoes some post-processing (as defined in Eq.~\eqref{eq:def_post-processing}) which generates $\widetilde{E}=\{\widetilde{E}_{a'}\}_{a'}$. It is easy to check that
\begin{align*}
    \Upsilon_p(\widetilde{E},F) &= \sum_{a' b} \| [\widetilde{E}_{a'},F_b]\|_p
    = \sum_{a' b} \|[\sum_a P_{a'|a}E_a,F_b]\|_p
    = \sum_{a' b} \|\sum_a P_{a'|a} [E_a,F_b]\|_p\\
    &\nctwo \sum_{a a' b} P_{a'|a} \|[E_a,F_b]\|_p
    \ncthree \sum_{a b} \|[E_a,F_b]\|_p
    = \Upsilon_p(E,F).
\end{align*}
Due to linearity of the commutator and non-negativity of the post-processing coefficients we can take the summation out of the commutator and then in (1) we use the triangle inequality for Schatten norms, which then in (2) allows us to perform an explicit summation over $a'$.
Clearly, this argument extends to the situation in which both measurements undergo some post-processing.

Let us now show that our measure fails to be monotonic under pre-processing.
Let $E$ and $F$ be a pair of two-outcome projective measurements on a qutrit:
\begin{align*}
    E_1 &= 
    \begin{bmatrix}
        1 & 0 & 0\\
        0 & 1 & 0\\
        0 & 0 & 0
    \end{bmatrix}, \hspace{0.5cm}
    E_2 =
    \begin{bmatrix}
        0 & 0 & 0\\
        0 & 0 & 0\\
        0 & 0 & 1
    \end{bmatrix},\\
    F_1 &=
    \begin{bmatrix}
        1 & 0 & 0\\
        0 & 0 & 0\\
        0 & 0 & 1
    \end{bmatrix}, \hspace{0.5cm}
    F_2 =
    \begin{bmatrix}
        0 & 0 & 0\\
        0 & 1 & 0\\
        0 & 0 & 0
    \end{bmatrix}.
\end{align*}
Since the measurement operators commute, we immediately see that $\Upsilon_p (E,F) = 0$.

Now consider a unital map $\Lambda^\dagger : L(\mathbb{C}^3) \xrightarrow{} L(\mathbb{C}^2)$ given by
\begin{align*}
    \Lambda^\dagger (A) = \sum_j K_j^\dagger A K_j,
\end{align*}
where
\begin{align*}
    K_1 = \frac{\sqrt{2}}{3}
    \begin{bmatrix}
        1 & 0\\
        1 & 0\\
        -1 & 0
    \end{bmatrix}, \hspace{0.5cm}
    K_2 =
    \begin{bmatrix}
        -\frac{1}{2\sqrt{3}} & \frac{1}{2}\\
        \frac{1}{2\sqrt{3}} & -\frac{1}{2}\\
        -0 & 0
    \end{bmatrix}, \hspace{0.5cm}
    K_3 =
    \begin{bmatrix}
        -\frac{1}{\sqrt{6}} & -\frac{1}{\sqrt{2}}\\
        0 & 0\\
        0 & 0
    \end{bmatrix}.
\end{align*}
It is easy to calculate that $\Upsilon_p (\Lambda^\dagger (E),\Lambda^\dagger (F)) =  \frac{2^{\frac{1}{p}+2}}{9\sqrt{3}}$. The fact that compatible measurements become incompatible after pre-processing might seem puzzling at first, but what it really shows is that commutativity is a fragile property which is easily disturbed.

One might ask whether the counterexample above is the simplest one. Since both measurements have two outcomes and the output system is a qubit, the only part we can hope to simplify is the dimension of the input system. It turns out that for qubit-to-qubit pre-processing a situation in which compatible measurements become incompatible is not possible. An arbitrary $2 \times 2$ operator $A$ can be written as
\begin{align*}
    A = c_A \mathbb{1} + T,
\end{align*}
where $c_{A}$ is a constant and $T$ is a linear combination of the Pauli matrices $\sigma_{x}, \sigma_{y}, \sigma_{z}$. Every operator $B$ which commutes with $A$ must necessarily be of the form
\begin{align*}
    B = c_B \mathbb{1} + \alpha T
\end{align*}
for some constants $c_{B}$ and $\alpha$. It is now clear that for any unital map $\Lambda^{\dagger}$, we have $[\Lambda^{\dagger}(A), \Lambda^{\dagger}(B) ] = 0$.

Note that the argument above does not rule out a weaker failure of monotonicity in which an incompatible pair of qubit measurements becomes more incompatible under pre-processing. However, we have not found such an example during our numerical investigation.

\subsection{The most incompatible measurements}
In this section we introduce the notion of the most incompatible pairs of measurements according to our measures. We also provide a complete characterization of such pairs.

Let $\Upsilon_p^*(n_E,n_F,d)$ be the largest value of the incompatibility measure $\Upsilon_p$ achievable by a pair of measurements with $n_E$ and $n_F$ outcomes, respectively, acting on dimension $d$:
\begin{align*}
    \Upsilon_p^*(n_E,n_F,d):= \sup_{E,F} \Upsilon_p(E,F).
\end{align*}
Since the set of measurements with a fixed number of outcomes acting on a fixed dimension is compact and the measures are continuous, the supremum is actually achieved, i.e.
\begin{align*}
    \Upsilon_p^*(n_E,n_F,d) = \max_{E,F} \Upsilon_p(E,F).
\end{align*}

It is well-known that every measurement can be obtained by post-processing a rank-1 measurement. This is important, because for pairs of rank-1 measurements we can derive a tight upper bound on the incompatibility measure.
\begin{lem}
\label{lem:most-incompatible}
Let $E$ and $F$ be two measurements acting on $\mathbb{C}^{d}$.
Then, $\Upsilon_{p}(E, F) \leq 2^\frac{1}{p} d\sqrt{d-1}$.
This upper bound is tight as it is saturated by a pair of mutually unbiased bases in $\mathbb{C}^{d}$.
\end{lem}
\begin{proof}
In Section~\ref{sec:post-processing} we have shown that our measure is non-increasing under post-processing, so for the purpose of deriving upper bounds it suffices to consider pairs of (not necessarily projective) rank-1 measurements $E=\{E_a\}_{a=1}^{n_E}$ and $F=\{F_b\}_{b=1}^{n_F}$ in $\mathbb{C}^d$. Then, we can write the measurement operators as: $E_a = \alpha_a \ket{e_a} \bra{e_a}$, $F_b = \beta_b \ket{f_b} \bra{f_b}$, where $\alpha_a, \beta_b \in [0,1]$ and $c_{ab} := |\braket{e_a|f_b}|$.
Corollary~\ref{Corollary:norm_of_com} implies that for rank-1 measurements we have
\begin{align}
\label{eq:overlap_funtion_c_a_b}
    \Upsilon_p (E,F) = \sum_{a b} \alpha_a \beta_b h_p(c_{ab}),
\end{align}
where $h_{p}(c)$ is the function defined in Eq.~\eqref{eq:h_p_def}. In this case the measures corresponding to different $p \in [1,\infty]$ are equivalent and differ only by the numerical prefactor.

Normalization of each measurement implies that
\begin{align*}
    \sum_a \alpha_a = \sum_b \beta_b = d.
\end{align*}
Moreover, we have
\begin{align*}
    \tr( E_a \cdot F_b ) = \alpha_a \beta_b |c_{ab}|^2
\end{align*}
and by summing over $a$ and $b$ we obtain
\begin{align*}
    \sum_{a b} \alpha_a \beta_b c_{ab}^2 = d.
\end{align*}
To simplify our calculations we define
$t_{ab} := c_{ab}^2$ and $\widetilde{h}_p(t_{ab}) := 2^\frac{1}{p} \cdot \sqrt{t_{ab}(1-t_{ab})}$.
To see that $\widetilde{h}_p$ is strictly concave on the interval $[0, 1]$ note that it is a composition of a strictly concave quadratic function with the square root function, which is strictly increasing and strictly concave.
This allows us to bound the right-hand side of Eq.~\eqref{eq:overlap_funtion_c_a_b} using Jensen's inequality:
\begin{align*}
    \sum_{a b} \frac{\alpha_a \beta_b}{d^2} \widetilde{h}_p(t_{ab}) \leq \widetilde{h}_p \bigg( \sum_{a b} \frac{\alpha_a \beta_b}{d^2} t_{ab} \bigg) = \widetilde{h}_p \Big( \frac{1}{d} \Big) = 2^\frac{1}{p} \sqrt{\frac{1}{d} \left( 1 - \frac{1}{d} \right)},
\end{align*}
where the factor of $\frac{1}{d^{2}}$ is required for normalization.
Clearly, this upper bound is tight for a pair of MUBs in dimension $d$.
\end{proof}

\begin{cor}
\label{corollary:after_lem3}
Since $\widetilde{h}_p(t_{ab})$ is strictly concave on the interval $[0,1]$, the only pairs of rank-1 measurements which achieve the maximal value are those for which all the overlaps are equal to $\frac{1}{\sqrt{d}}$.
\end{cor}

The fact that the maximal value of $\Upsilon_p$ shows a strong dependence on the dimension suggests that we should not compare values for measurements acting on different dimensions.
In principle, we could introduce a dimension-dependent prefactor to ensure that the maximal value is always equal to $1$, but this would be rather arbitrary. 
This explains why we should not expect the measure to behave well in scenarios where the dimension of the quantum system is not preserved, e.g.~under pre-processing.

So far we have shown that for rank-1 measurements achieving the maximal incompatibility is equivalent to having uniform overlaps. It is then natural to ask whether maximal incompatibility can be achieved by measurements which are not rank-1 and in this section we show that this is not possible. Let us start by proving the following technical lemma. 
\begin{lem}
\label{lem:sharp_trig_ineq}
Let $\{\ket{0}, \ket{1}, \ket{2}\}$ be an orthonormal basis on $\mathbb{C}^{3}$. Consider the following operators:
\begin{align*}
    P_0 &= \ket{0} \bra{0},\\
    P_1 &= \gamma \ket{1} \bra{1}, \hspace{0.5cm} 
    \gamma \in \, ]0,1]\\
    Q &= \ket{\psi} \bra{\psi}, \hspace{0.5cm} \ket{\psi} = a \ket{0} + b \ket{1} + c \ket{2},\\
    A &= i [P_0,Q],\\
    B &= i [P_1,Q],
\end{align*}
where $a, b, c \in \mathbb{C}$ and $a b \neq 0$.
Then, for any $p \in [1, \infty]$ we have
\begin{align*}
    \| A + B \|_p < \| A \|_p + \| B \|_p.
\end{align*}
\end{lem}
\begin{proof}
Let us first consider the special case of $c=0$. It is easy to see that:
\begin{align*}
    A &= i( ab^* \ket{0} \bra{1} - a^* b \ket{1} \bra{0}),\\
    B &= i \gamma ( -ab^* \ket{0} \bra{1} + a^* b \ket{1} \bra{0} ),
\end{align*}
so $B = - \gamma A$. Then, the triangle inequality is necessarily strict because $2^\frac{1}{p} (1 - \gamma)|a|\cdot |b| < 2^\frac{1}{p} (1 + \gamma) |a|\cdot |b|$ for $\gamma\in \, ]0,1]$ and $a b \neq 0$.

In the case of $c \neq 0$ let us write the hermitian operators $A$ and $B$ in their spectral decomposition.
The eigenvalues of $A$ are given by
\begin{align*}
     \lambda &= \{ 0, \alpha, -\alpha \}, \text{ for } \alpha = \sqrt{|a|^2 (|b|^2 + |c|^2)},
\end{align*}
and the corresponding unnormalized eigenvectors are
\begin{align*}
    \left\{ \Colvec{0,-\frac{c^*}{b^*},1}, \Colvec{\frac{i \alpha}{a^* c}, \frac{b}{c}, 1}, \Colvec{\frac{-i \alpha}{a^* c}, \frac{b}{c}, 1} \right\}.
\end{align*}
The eigenvalues of $B$ are given by
\begin{align*}
    \{ 0,\gamma \beta, -\gamma \beta \}, \text{ for } \beta = \sqrt{|b|^2 (|a|^2 + |c|^2)},
\end{align*}
and the corresponding unnormalized eigenvectors are
\begin{align*}
    \left\{ \Colvec{- \frac{c^*}{a^*}, 0, 1}, \Colvec{\frac{a}{c}, \frac{i \beta}{b^* c}, 1}, \Colvec{\frac{a}{c}, \frac{-i \beta}{b^* c}, 1} \right\}.
\end{align*}
Since $ab \neq 0$ immediately implies $\alpha, \beta > 0$, the positive and negative parts of $A$ and $B$ are 1-dimensional, so the operators can be written as:
\begin{align*}
    A &= \underbrace{\alpha P^+}_{=A^+} \underbrace{- \alpha P^-}_{=A^-},\\
    B &= \underbrace{\gamma \beta Q^+}_{=B^+} \underbrace{- \gamma \beta Q^-}_{=B^-},
\end{align*}
where $P^{\pm}$ and $Q^{\pm}$ are rank-1 projectors onto the corresponding eigenspaces.
To show that the triangle inequality is always strict, we use the equivalence proved in Lemma~\ref{lem:tri_ineq_sat}, which reduces the problem to arguing that there does not exist a hermitian operator $X$ which is a simultaneous optimizer for both $A$ and $B$. 
Let us first apply Lemma~\ref{lem1} to both $A$ and $B$ to restrict the set of potential optimizers.
More specifically, it requires that the positive part of $X$ is orthogonal to the negative parts of both $A$ and $B$ and similarly the negative part of $X$ is orthogonal to the positive parts of $A$ and $B$.
Let us first argue that $A^{+}$ and $B^{+}$ do not project on the same subspace. To do so it suffices to show that their eigenvectors are not proportional to each other. Since the last coordinate of both vectors is equal to $1$, it suffices to check whether the other coordinates are equal. Equating the first coordinates leads to
\begin{align*}
    \frac{i \alpha}{a^* c} &= \frac{a}{c},
\end{align*}
from which we conclude that $i \alpha = |a|^2$. However, $\alpha$ is a real positive number, so this is a contradiction.
Since $A^+$ and $B^+$ are $1$-dimensional and not proportional to each other, any positive semi-definite operator acting on $\mathbb{C}^{3}$ which is orthogonal to both of them must be of rank at most 1, i.e.~$\rank(X^{-}) \leq 1$.
An analogous argument holds for the negative parts of $A$ and $B$, which implies that the positive part of the potential optimizer, $X^{+}$, has rank at most 1.
This leaves us with two possible forms for $X$: (a) either $X$ has two non-zero eigenvalues (1 positive and 1 negative) or (b) it has only 1 non-zero eigenvalue. In the rest of the proof we denote the projectors on the positive and negative part of $X$ by $R^\pm$.

Let us first show that there does not exist an optimal $X$ of the first kind, i.e.~having 1 positive and 1 negative eigenvalue. Then, $R^+ = \ket{\pi^{+}} \bra{\pi^{+}}$ and $R^- = \ket{\pi^{-}} \bra{\pi^{-}}$, where $\ket{\pi^{\pm}}$ are normalized eigenvectors of $X$.
Since the potential eigenvectors of $X$ are uniquely defined through orthogonality to $A^{\pm}$ and $B^{\pm}$ we can write them down (up to a normalisation constant):
\begin{align*}
    \ket{\pi^+} \propto \Colvec{\frac{b^*}{c^*} - \frac{i \beta}{b c^*}, \frac{a^*}{c^*} - \frac{i \alpha}{a c^*}, - \frac{a^* b^*}{(c^*)^2} - \frac{\alpha \beta}{a b (c^*)^2}},\\
    \ket{\pi^-} \propto \Colvec{\frac{b^*}{c^*} + \frac{i \beta}{b c^*}, \frac{a^*}{c^*} + \frac{i \alpha}{a c^*}, - \frac{a^* b^*}{(c^*)^2} - \frac{\alpha \beta}{a b (c^*)^2}}.
\end{align*}
Then, it is easy to see that
\begin{align*}
    \braket{\pi^+|\pi^-} = \frac{1}{\mathcal{N}} \left[ |a|^2 |b|^2 |c|^2 (|b|^2 + |a|^2) + (\alpha \beta + |a|^2 |b|^2)^2 - |c|^2 (|b|^2 \alpha^2 + |a|^2 \beta^2) + 2i |a|^2 |b|^2 |c|^2 (\alpha + \beta) \right],
\end{align*}
where $\mathcal{N}$ is some non-zero normalization constant.
As we can see, for non-zero coefficients $a, b, c$ we always have a non-vanishing imaginary part, which implies that the vectors are not orthogonal. 
This rules out the possibility of $X$ having two non-zero eigenvalues, because the positive and negative parts of a hermitian operator must be orthogonal.

The second option we have to consider corresponds to:
\begin{align*}
    X = R^+ \; \text{or} \; X = - R^-.
\end{align*}
Let us start with the case $X = R^+$. By calculating the inner product of $A$ and $X$ we obtain:
\begin{align*}
    \langle A, X \rangle = \underbrace{\langle A^+, R^+ \rangle}_{\geq 0} + \underbrace{\langle A^-, R^+ \rangle}_{\leq 0} \leq \langle A^+, R^+ \rangle = \langle \alpha P^+, R^+ \rangle = \alpha \tr(P^+ R^+) \leq \alpha.
\end{align*}
Since $P^{+}$ and $R^{+}$ are rank-1 projectors, in order to saturate the final inequality we require $P^+ = R^+$. 
From a similar analysis of $\langle B, X \rangle$, we conclude that $Q^+ = R^+$. 
However, these conditions together imply that that $P^+ = Q^+$, which we have shown before to be false. 
Identical result can be obtained in the case $X = - R^-$.
Finally, this proves that there does not exist $X$ which is optimal for both $A$ and $B$ simultaneously.
\end{proof}

We are now ready to prove the main result of this section.

\begin{lem}
\label{lem:app_lem_sharp_trig_ineq}
For a pair of measurements $E=\{E_a\}_{a=1}^{n_{E}}$ and $F=\{F_b\}_{b=1}^{n_{F}}$, where at least one measurement operator is of rank strictly larger than 1, the maximum incompatibility cannot be achieved.
\end{lem}
\begin{proof}
Let us first write the measurement operators in their spectral decomposition:
\begin{align*}
    E_a &= \sum_{j=1}^{\rank(E_a)} \lambda_{a j} \ket{e_{aj}} \bra{e_{aj}},\\
    F_b &= \sum_{k=1}^{\rank(F_b)} \mu_{b k} \ket{f_{bk}} \bra{f_{bk}}.
\end{align*}
Without loss of generality we assume that $\rank(E_1) \geq 2$. 
Let us now assume that $E$ and $F$ achieve the maximal incompatiblity and show that this leads to a contradiction.
If we decompose all the measurement operators to be rank-1, we obtain:
\begin{align*}
    G_{aj} &= \lambda_{a j} \ket{e_{aj}} \bra{e_{aj}},\\
    H_{bk} &= \mu_{b k} \ket{f_{bk}} \bra{f_{bk}},
\end{align*}
Due to post-processing monotonicity this new pair of measurements must still achieve the maximum incompatibility. This means that all the overlaps must be equal: $| \braket{e_{aj}|f_{bk}} | = \frac{1}{\sqrt{d}}$.
Let us now apply a post-processing to the $G$ measurement in which we combine the first two elements of $G$ and leave the rest unchanged i.e.~$\widetilde{G}_{11} = G_{11} + G_{12}$.
We will now analyse how this affects the incompatibility achieved by these two operators with an arbitrary operator from the $H$ measurement, say $H_{11}$.
Similarly to Lemma~\ref{lem:sharp_trig_ineq} we have three rank-1 operators: $G_{11}$, $G_{12}$ and $H_{11}$, so without loss of generality we can assume that they act on $\mathbb{C}^3$. 
The operators $G_{11}$ and $G_{12}$ are orthogonal, since they arise from the spectral decomposition of $E_{1}$. 
Finally, $H_{11}$ is not proportional to either $G_{11}$ or $G_{12}$, because of the overlap condition. 
Therefore, all the conditions of Lemma~\ref{lem:sharp_trig_ineq} are satisfied and we conclude that
\begin{align*}
    \|[\widetilde{G}_{11},H_{11}] \|_p < \|[G_{11},H_{11}] \|_p + \|[G_{12},H_{11}] \|_p.
\end{align*}
This implies that the value of our measure has necessarily decreased as $G$ got post-processed to $\widetilde{G}$. Since we can still obtain the original measurements $E$ and $F$ by applying further post-processing to $\widetilde{G}$ and $H$, this contradicts the initial assumption that $E$ and $F$ were maximally incompatible.
\end{proof}
\noindent Finally, combining Corollary~\ref{corollary:after_lem3} and Lemma~\ref{lem:app_lem_sharp_trig_ineq} leads to the following conclusion.
\begin{cor}
\label{Corollary:most_incompatible_measurements}
Let E and F be measurements on $\mathbb{C}^d$, such that $\Upsilon_p(E,F) = 2^\frac{1}{p} d\sqrt{d-1}$.
Then, all the measurements operators of $E$ and $F$ are rank-1 and their overlaps are equal to $\frac{1}{\sqrt{d}}$.
\end{cor}

\subsection{Composition of measurements}
Let us conclude this section by investigating the behavior of our measures under various types of compositions. We first describe different types of compositions using a single measurement as an example and then we apply it to pairs of measurements to see how the incompatibility measures are affected.

Given a measurement $\{ E_a \}_{a=1}^{n_{E}}$ acting on $\mathbb{C}^{d_{1}}$ a \textbf{trivial extension} corresponds to taking a tensor product with a finite-dimensional identity, i.e.~the resulting measurement $E'$ is given by:
\begin{align*}
    E'_a := E_a \otimes \mathbb{1}_{d_2}.
\end{align*}
Note that the number of outcomes of $E'$ is the same as that of $E$.

Given a pair of measurements $\{ E_a \}_{a=1}^{n_{E}}$ and $\{ \bar{E}_a \}_{a=1}^{n_{\bar{E}}}$ acting on $\mathbb{C}^{d_1}$ and $\mathbb{C}^{d_2}$, respectively, a \textbf{direct sum composition} yields a measurement acting on $\mathbb{C}^{d_1} \oplus \mathbb{C}^{d_2}$, whose measurement operators are given by:
\begin{align*}
    E'_{a} =
    \begin{cases}
    E_a \oplus 0 &\nbox{for} a \leq n_{E},\\
    0 \oplus \bar{E}_{a - n_{E}} &\nbox{for} a > n_{E},
    \end{cases}
\end{align*}
where $a \in \{1, 2, \ldots, n_{E} + n_{\bar{E}}\}$.

Given a pair of measurements $\{ E_a \}_{a=1}^{n_{E}}$ and $\{ \bar{E}_{\bar{a}} \}_{\bar{a}=1}^{n_{\bar{E}}}$ acting on $\mathbb{C}^{d_1}$ and $\mathbb{C}^{d_2}$, respectively, a \textbf{tensor product composition} yields a measurement acting on $\mathbb{C}^{d_1} \otimes \mathbb{C}^{d_2}$, whose measurement operators are given by:
\begin{align*}
    E'_{ a \bar{a} } := E_{a} \otimes \bar{E}_{\bar{a}}.
\end{align*}
Note that $E'$ has $n_{E} n_{\bar{E}}$ outcomes which are labelled by pairs $(a, \bar{a})$. Let us now investigate how these transformations affect the incompatibility measures.

In the case of a \textbf{trivial extension} the measurements $\{ E_a \}_{a=1}^{n_{E}}$ and $\{ F_b \}_{b=1}^{n_{F}}$ acting on $\mathbb{C}^{d_1}$
get transformed to $\{ E'_a \}_{a=1}^{n_{E}}$ and $\{ F'_b \}_{b=1}^{n_{F}}$, where
\begin{align*}
    E'_a := E_a \otimes \mathbb{1}_{d_2} \nbox{and} F'_{b} := F_b \otimes \mathbb{1}_{d_2}.
\end{align*}
Since
\begin{align*}
    [E_a \otimes \mathbb{1}_{d_2},F_b \otimes \mathbb{1}_{d_2}] = [E_a,F_b] \otimes \mathbb{1}_{d_2},
\end{align*}
we immediately conclude that
\begin{align*}
    \Upsilon_p (E', F') &= \sum_{a b} \||[E_a \otimes \mathbb{1}_{d_2},F_b \otimes \mathbb{1}_{d_2}]\|_p\\
    &= \sum_{a b} \|[E_a,F_b]\|_p \cdot \|\mathbb{1}_{d_2}\|_p\\
    &= d_2^\frac{1}{p} \cdot \Upsilon_p (E, F).
\end{align*}
Clearly, such a trivial extension gives rise to a multiplicative factor in the value of the measure.

In the case of a \textbf{direct sum composition} we have measurements $\{ E_a \}_{a=1}^{n_{E}}$ and $\{ F_b \}_{b=1}^{n_{F}}$ acting on $\mathbb{C}^{d_1}$ and measurements $\{ \bar{E}_a \}_{a=1}^{n_{\bar{E}}}$ and $\{ \bar{F}_b \}_{b=1}^{n_{\bar{F}}}$ acting on $\mathbb{C}^{d_2}$. These two pairs get transformed into $E'$ and $F'$ given by:
\begin{align*}
    E'_{a} &=
        \begin{cases}
            E_a \oplus 0 &\nbox{for} a \leq n_{E},\\
            0 \oplus \bar{E}_{a - n_{E}} &\nbox{for} a > n_{E},
        \end{cases} \\
    F'_{b} &=
        \begin{cases}
            F_b \oplus 0 &\nbox{for} b \leq n_{F},\\
            0 \oplus \bar{F}_{b - n_{F}} &\nbox{for} a > n_{F}.
        \end{cases}
\end{align*}
To compute the commutator we need to consider four distinct cases depending on whether $a \leq n_{E}$ or $a > n_{E}$ and $b \leq n_{F}$ or $b > n_{F}$. It is clear that in two out of the four possible cases the commutator vanishes and we only get a non-trivial contribution if either (1) $a \leq n_{E}$ and $b \leq n_{F}$ or (2) $a > n_{E}$ and $b > n_{F}$. Therefore, we have
\begin{align*}
    \Upsilon_p (E',F') = 
    \sum_{a = 1}^{n_{E} + 
    n_{\bar{E}}} \sum_{b = 1}^{ n_{F} + 
    n_{\bar{F}} } \| [E'_{a},F'_{b}] \|_p =
   \sum_{a = 1}^{n_{E}} \sum_{b = 1}^{n_{F}} \| [E_a,F_b] \|_p +
    \sum_{a = 1}^{n_{\bar{E}}} \sum_{b = 1}^{n_{\bar{F}}} \| [\bar{E}_a,\bar{F}_b] \|_p =
    \Upsilon_p (E,F) + \Upsilon_p (\bar{E}, \bar{F}),
\end{align*}
which means that our measures are additive under direct sum composition.

Finally, in the case of a \textbf{tensor product composition} we have measurements $\{ E_a \}_{a=1}^{n_{E}}$ and $\{ F_b \}_{b=1}^{n_{F}}$ acting on $\mathbb{C}^{d_1}$ and $\{ \bar{E}_{\bar{a}} \}_{\bar{a}=1}^{n_{\bar{E}}}$ and $\{ \bar{F}_{\bar{b}} \}_{\bar{b}=1}^{n_{\bar{F}}}$ acting on $\mathbb{C}^{d_2}$. These two pairs get transformed into $E'$ and $F'$ given by
\begin{align*}
    E'_{ a \bar{a} } := E_a \otimes \bar{E}_{\bar{a}} \nbox{and} F'_{b \bar{b}} := F_b \otimes \bar{F}_{\bar{b}}.
\end{align*}
Note that the measurement $E'$ and $F'$ have $n_{E} n_{\bar{E}}$ and $n_{F} n_{\bar{F}}$ outcomes, respectively. 
Since the commutator $[ E'_{a \bar{a}}, F'_{b \bar{b}} ]$ cannot be expressed through the original commutators $[ E_{a}, F_{b} ]$ and $[ \bar{E}_{\bar{a}}, \bar{F}_{\bar{b}} ]$, it is hard to study the behavior of the measure for general quantum measurements.
Under the assumption that the original measurements are rank-1, which implies that both $E'$ and $F'$ are rank-1, the value of $\Upsilon_{p} (E', F')$ can be written as a relatively simple function of overlaps and traces, but this function cannot be expressed in terms of $\Upsilon_{p} (E, F)$ and $\Upsilon_{p} (\bar{E}, \bar{F})$ alone.
To see this suppose that all the measurements are rank-1 projective and the overlaps between $E$ and $F$ are denoted by $c_{ab}$ and similarly the overlaps between $\bar{E}$ and $\bar{F}$ are denoted by $\bar{c}_{\bar{a} \bar{b}}$.
Then, it is easy to check that
\begin{align*}
    \Upsilon_p (E', F') = 2^\frac{1}{p} \sum_{a \bar{a} b \bar{b} } c_{ab} \bar{c}_{\bar{a}\bar{b}} \sqrt{1 - (c_{ab} \bar{c}_{\bar{a}\bar{b}})^2}.
\end{align*}
Clearly, it is the term under the square root which prevents us from factorising the entire expression.
A special case where this expression can be evaluated analytically is when $E$ and $F$ correspond to MUBs in dimension $d_{1}$, while $\bar{E}$ and $\bar{F}$ correspond to MUBs in dimension $d_{2}$. Then,
\begin{align*}
    \Upsilon_p (E', F') = 2^\frac{1}{p} \cdot (d_1 d_2)^2 \sqrt{\frac{1}{d_1 d_2} \left( 1 - \frac{1}{d_1 d_2} \right)} = 2^\frac{1}{p} \cdot d_1 d_2 \sqrt{d_1 d_2 - 1}
\end{align*}
and the correctness of the final value results from the fact that $E'$ and $F'$ constitute MUBs in dimension $d = d_{1} d_{2}$.

\section{Comparison with existing results}
\label{sec:comparison_exist_results}
In this section we show how our incompatibility measures can be linked to some existing results, namely: a robustness-based incompatibility measure, QRAC performance and entropic uncertainty relations.

\subsection{Comparison with a robustness-based incompatibility measure}
Incompatibility measures based on robustness to noise are widely used in the literature, but only one of them can be proven to be optimized by MUBs. Following the notation from Ref.~\cite{DFK19} incompatibility generalized robustness of a pair of measurements $\{E_{a}\}_{a = 1}^{n_{E}}$ and $\{F_{b}\}_{b = 1}^{n_{F}}$ is defined through the following semi-definite program:
\begin{align*}
    \eta^{g} :=
    \begin{cases}
        \max\limits_{\eta, \{ G_{ab} \}_{ab}} \; &\eta\\
        \hspace{0.4cm} \text{s.t.} \; &G_{ab} \geq 0,\\
        &\sum_{b} G_{ab} \geq \eta E_a,\\
        &\sum_a G_{ab} \geq \eta F_b,\\
        &\sum_{ab} G_{ab} = \mathbb{1}.
    \end{cases}
\end{align*}
Moreover, it has been proven in Ref.~\cite{DFK19} that for any pair of measurements acting on $\mathbb{C}^{d}$, we have
\begin{align*}
    \eta^{g} \geq \frac{1}{2} \left( 1 + \frac{1}{\sqrt{d}} \right).
\end{align*}
Note that due to the definition of robustness-based measures saturating this lower bound characterizes the most incompatible measurements. An upper bound on $\eta^{g}$ can be obtained by providing a solution to the dual problem given by:
\begin{align*}
    \begin{cases}
        \min\limits_{N,\{X_a\}_a, \{Y_b\}_b} \; &\tr(N)\\
        \hspace{0.4cm} \text{s.t.} \; &N=N^\dagger,\\
        &N \geq X_a + Y_b,\\
        &X_a \geq 0, \; Y_b \geq 0,\\
        &\sum_a \tr(X_a E_a ) + \sum_b \tr (Y_b F_b) \geq 1.
    \end{cases}
\end{align*}
We are now ready to show that any pair of measurements that achieves the  maximal incompatibility according to the commutation-based measures achieves the lowest possible value of $\eta^{g}$. Note that a related statement which applies only to rank-1 measurements acting on a qubit can be found in Section 4.2 of Ref.~\cite{DFK19}.

Corollary~\ref{Corollary:most_incompatible_measurements} tells us that a pair of measurements which is maximally incompatible with respect to the commutation-based measures must consist of rank-1 measurements whose overlaps are equal to $\frac{1}{\sqrt{d}}$. Then, we consider the following assignment of the dual variables:
\begin{align*}
    X_a &:= \frac{1}{2d} \cdot \frac{E_a}{\tr (E_a)},\\
    Y_b &:= \frac{1}{2d} \cdot \frac{F_b}{\tr (F_b)}.
\end{align*}
Using the fact that for rank-1 operators the square of the trace equals the trace of the square, we can immediately verify that
\begin{align*}
    \frac{1}{2d} \bigg( \sum_a  \frac{\tr (E_a^2)}{\tr (E_a)} + \sum_b \frac{\tr (F_b^2)}{\tr (F_b)} \bigg) = \frac{1}{2d} \bigg( \sum_a  \tr (E_a) + \sum_b \tr (F_b) \bigg) = 1.
\end{align*}
Finally, we set $N := \frac{1}{2d} ( 1 + \frac{1}{\sqrt{d}}) \mathbb{1}$, so to prove the operator inequality $N \geq X_{a} + Y_{b}$ it suffices to show that
\begin{align*}
\| X_a + Y_b \|_\infty \leq \frac{1}{2d} ( 1 + \frac{1}{\sqrt{d}})
\end{align*}
for all $a, b$. However, since $X_{a} + Y_{b}$ is always a rank-2 operator we can use Lemma~\ref{lem:spec_commutator} to find its eigenvalues and explicitly check that the condition stated above is satisfied. It is then easy to see that the resulting upper bound on $\eta^{g}$ coincides with the universal lower bound, which concludes the proof.

It is natural to ask whether the implication holds in the other direction, i.e.~whether every pair of measurements which is maximally incompatible according to $\eta^{g}$ is also maximally incompatible with respect to commutation-based measures. This, however, we cannot answer because we do not have a full characterization of the set of maximally incompatible pairs with respect to $\eta^{g}$. 

\subsection{Relation to the QRAC performance}
\label{sec:QRAC_performance}
In this section we investigate the relation between commutation-based incompatibility measures and the QRAC performance. More specifically, our goal is to show that if a pair of measurements is capable of producing a nearly-optimal performance in the standard QRAC, then it must be close to maximally incompatible with respect to the commutation-based measures. Since the relation between incompatibility and QRAC performance for the most general measurements turns out to be rather difficult to study, we restrict our attention to the case of rank-1 projective measurements.

Consider the standard QRAC in which two classical dits are encoded in a single qudit. Since our goal is to capture the usefulness of measurements in a QRAC, we assume that the preparations are chosen optimally. Then, the average success probability of a QRAC (see e.g.~Ref.~\cite{def_prob_qrac} for a definition), denoted by $\pave$, is given by
\begin{equation}
\label{eq:pave}
\pave = \frac{1}{2 d^{2}} \sum_{a b} \| E_a + F_b \|_\infty = \frac{1}{2} + \frac{1}{2 d^{2}} \sum_{ab} c_{ab},
\end{equation}
where the final equality results from the fact that for rank-1 projective measurements the norm can be calculated explicitly and is a simple function of the overlaps, which we denote by $c_{ab}$~\cite{PhysRevLett.121.050501}. Therefore, knowing how useful a pair of rank-1 measurements is in a QRAC is equivalent to knowing the value of the sum $\sum_{ab} c_{ab}$. Let us now show that this is sufficient to prove a non-trivial lower bound on the commutation-based incompatibility measures.

The argument proceeds in two steps. We first show that given the value of $\sum_{ab} c_{ab}$, we can bound the variance of the overlaps and then we use the fact that the function appearing in the commutation-based measure is bounded and concave to relate the variation of the function to the variation of the arguments.

While the overlaps $c_{ab}$ are labelled my two integers in the range $\{1, 2, \ldots, d\}$ in our argument we ignore this structure and we think of them as a set of $d^{2}$ numbers, so we will simply write $c_{j}$, where the summation goes over $j \in \{1, 2, \ldots, d^{2}\}$. Let us first define the average overlap:
\begin{align*}
    \bar{c} := \frac{1}{d^2} \sum_{j} c_j.
\end{align*}
The standard relation between vector $p$-norms implies that
\begin{equation}
\label{eq:norm-relation}
\sum_{j} |c_j - \bar{c}| \leq d \sqrt{ \sum_{j} ( c_j - \bar{c} )^{2} }
\end{equation}
and note that the expression under the square root can be evaluated to give:
\begin{align*}
\sum_{j} (c_j - \bar{c})^2 = \sum_{j} c_j^2 - 2 \bar{c} \sum_{j} c_j + d^2 \bar{c}^2 = \sum_{j} c_j^2 - d^2 \bar{c}^2 = d - d^2 \bar{c}^2,
\end{align*}
where we have used the fact that the overlaps satisfy $\sum_{j} c_{j}^{2} = d$. Therefore, knowing the average overlap $\bar{c}$ gives us an upper bound on the sum $\sum_{j} |c_j - \bar{c}|$.

According to Eq.~\eqref{eq:overlap_funtion_c_a_b} in the case of rank-1 projective measurements the incompatibility measure is given by:
\begin{align*}
\Upsilon_p (E,F) = \sum_{j} h_{p}(c_{j}),
\end{align*}
which can be rewritten as
\begin{equation}
\label{eq:Upsilon-rewritten}
\Upsilon_p (E,F) = d^{2} h_{p}( \bar{c} ) + \sum_{j} h_{p}(c_{j}) - h_p(\bar{c}).
\end{equation}
Our goal now is to provide a lower bound on the second term and we do this by thinking of $c_{j}$ as a deviation from $\bar{c}$. Since $h_{p}$ is a concave function we can obtain an upper bound by considering its derivative:
\begin{equation}
\label{eq:upper-bound}
h_{p}( c_{j} ) - h_{p}(\bar{c}) \leq h_{p}'( \bar{c} ) ( c_{j} - \bar{c} ),
\end{equation}
where $h_{p}'( \bar{c} )$ is the derivative of $h_{p}$ evaluated at $\bar{c}$.

To bound the same expression from below we use the fact that the domain is a bounded interval, namely $[0, 1]$. Therefore, it suffices to construct straight lines joining the point $\big( \bar{c}, h_{p}( \bar{c} ) \big)$ with the two endpoints, namely: $\big(0, h_{p}(0) \big)$ and $\big(1, h_{p}(1) \big)$. This implies that
\begin{equation}
\label{eq:lower-bound}
h_{p}( c_{j} ) - h_{p}(\bar{c}) \geq
\begin{cases}
\frac{ - h_{p}(\bar{c}) }{ 1 - \bar{c} } ( c_{j} - \bar{c} ) &\nbox{if} c_{j} \geq \bar{c},\\
\frac{ h_{p}(\bar{c}) }{\bar{c}} ( c_{j} - \bar{c} ) &\nbox{if} c_{j} < \bar{c}.
\end{cases}
\end{equation}
Geometrical justification of the bounds given in Eqs.~\eqref{eq:upper-bound} and~\eqref{eq:lower-bound} can be found in Fig.~\ref{fig:geometrical-justification}.
\begin{figure}[h!]
    \centering
    \includegraphics{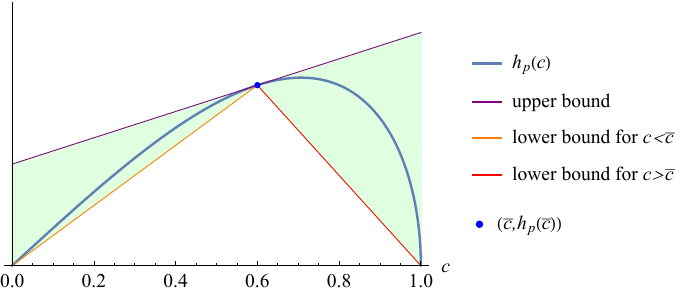}
    \caption{For every value of $\bar{c}$ the function $h_{p}(c)$ is bounded from above by its derivative at $c = \bar{c}$ and from below by straight lines connecting the point $\big( \bar{c}, h_{p}(\bar{c}) \big)$ to the endpoints, namely $(0, 0)$ and $(1, 0)$.}
    \label{fig:geometrical-justification}
\end{figure}

Combining the two cases of Eq.~\eqref{eq:lower-bound} into a single bound leads to
\begin{align*}
h_{p}( c_{j} ) - h_{p}(\bar{c}) \geq \max \bigg\{ \frac{h_{p}(\bar{c}) }{ 1 - \bar{c} }, \frac{h_{p}(\bar{c}) }{\bar{c}} \bigg\} \cdot | c_{j} - \bar{c} |
\end{align*}
and note that both terms appearing inside the maximum are non-negative. Together with Eq.~\eqref{eq:upper-bound} this implies that
\begin{align*}
    \big| h_{p}( c_{j} ) - h_{p}(\bar{c}) \big| \leq \alpha( \bar{c} ) \cdot |c_j - \bar{c}|
    \nbox{for} \alpha( \bar{c} ) := \max \bigg\{ \frac{h_{p}(\bar{c}) }{ 1 - \bar{c} }, \frac{h_{p}(\bar{c}) }{\bar{c}}, |h_{p}'( \bar{c} )| \bigg\}.
\end{align*}
Plugging this bound into Eq.~\eqref{eq:Upsilon-rewritten} gives
\begin{align*}
\Upsilon_p (E,F) &= d^{2} h_{p}( \bar{c} ) + \sum_{j} h_{p}(c_{j}) - h_p(\bar{c}) \geq d^{2} h_{p}( \bar{c} ) - \sum_{j} \big| h_{p}(c_{j}) - h_p(\bar{c}) \big|\\
&\geq d^{2} h_{p}( \bar{c} ) - \alpha( \bar{c} ) \sum_{j} | c_{j} - \bar{c} | \geq d^{2} h_{p}( \bar{c} ) - \alpha( \bar{c} ) d \sqrt{ d - d^{2} \bar{c}^{2} } =: f(p, \bar{c},d),
\end{align*}
where the last inequality comes from the norm relation given in Eq.~\eqref{eq:norm-relation}.

Thanks to Eq.~\eqref{eq:pave} the average overlap is directly linked to the QRAC performance: $\bar{c} = 2 \pave - 1$. Since the relevant range of $\pave$ is $[ \frac{1}{2} + \frac{1}{2d}, \frac{1}{2} + \frac{1}{2 \sqrt{d}} ]$, the corresponding range of $\bar{c}$ is $[ 1/d, 1/\sqrt{d}]$.

Let us first show that our analysis is tight in the case of the optimal QRAC performance. If $\pave = \frac{1}{2} + \frac{1}{2 \sqrt{d}}$, then $\bar{c} = \frac{1}{\sqrt{d}}$ and since $\alpha(\bar{c})$ is finite we immediately deduce that
\begin{equation*}
\Upsilon_p (E,F) \geq d^{2} h_{p}( \bar{c} ) = 2^\frac{1}{p} d\sqrt{d-1},
\end{equation*}
which coincides with the maximal value of incompatibility in dimension $d$.

For suboptimal QRAC performance an analytical bound is easy to compute. In Fig.~\ref{fig:qrac-incompatibility} we present plots of $f(1, \bar{c}, 2)$ and $f(1, \bar{c}, 3)$ in the relevant range of $\bar{c}$.
\begin{figure}[h!]
    \centering
    \includegraphics{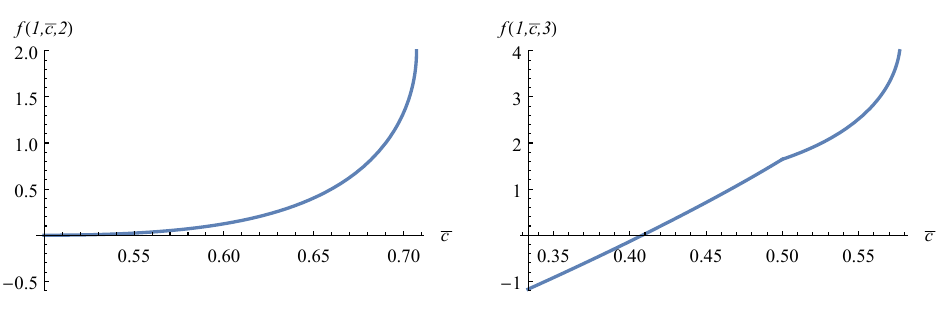}
    \caption{Lower bounds on the commutation-based incompatibility measure as a function of the average overlap for $d = 2$ (left) and $d = 3$ (right). Note that for $d = 2$ we obtain a non-trivial bound for the entire range of $\bar{c}$, but this is not the case for $d = 3$. As showed before in both cases the optimal QRAC performance certifies the maximal value of incompatibility.}
    \label{fig:qrac-incompatibility}
\end{figure}

\subsection{Relation to entropic uncertainty}
In this section we take a closer look at the relation between incompatibility measures and entropic uncertainty relations. More specifically, we will calculate an explicit upper and lower bound on incompatibility for rank-1 projective measurements as a function of uncertainty. This will allow us to conclude that certain regions in the incompatibility--uncertainty space are forbidden.

Let $E$ and $F$ be rank-1 projective measurements acting on $\mathbb{C}^{d}$, i.e.~for $a, b \in \{1, 2, \ldots, d\}$ we have
\begin{align*}
    E_a &= \ket{ e_{a} } \bra{ e_{a} },\\
    F_b &= \ket{ f_{b} } \bra{ f_{b} },
\end{align*}
where $\{ \ket{ e_{a} } \}_{a}$ and $\{ \ket{ f_{b} } \}_{b}$ are orthonormal bases on $\mathbb{C}^{d}$. Let $t_{ab} := | \braket{ e_{a} | f_{b} }|^2$ be the square of the overlap. The entropic uncertainty relation introduced by Maassen and Uffink in Ref.~\cite{PhysRevLett.60.1103} states that:
\begin{align*}
    H(E) + &H(F) \geq - \log \tau,\\
    \tau &:= \max_{a b} t_{a b},
\end{align*}
where $H(E)$ and $H(F)$ are the Shannon entropies of the probability distributions arising when measurements $E$ and $F$ are performed on some fixed quantum state. When $\tau = 1$, the right-hand side vanishes and the inequality becomes trivial. This happens only when the two bases share a vector and then no uncertainty can be guaranteed. The other extreme corresponds to $\tau = \frac{1}{d}$, which occurs only when the bases are mutually unbiased and then the right-hand side evaluates to $\log d$.

As the right-hand side does not depend on the state, this is an example of a state-independent uncertainty relation (it depends only on the measurements performed just like our incompatibility measures). From now on we will treat the right-hand side of the uncertainty relation (namely $- \log \tau$) as a measure of uncertainty. To explore the interplay between incompatibility and uncertainty, we will derive upper and lower bounds on incompatibility as a function of uncertainty.

Therefore, our task is to maximize/minimize incompatibility measures for a fixed value of $\tau$. Since both uncertainty and incompatibility are invariant under permutations of outcomes, we can without loss of generality assume that $\tau = t_{11}$ and the relevant range of $\tau$ is given by $[ \frac{1}{d}, 1 ]$. Note that if we define the matrix $T = (t_{ab})$:
\begin{align*}
    T = \begin{bmatrix}
    t_{11} & t_{21} & \cdots & t_{d1}\\
    t_{12} & t_{22} & & \vdots\\
    \vdots & & \ddots & \\
    t_{1d} & \cdots & & t_{dd}
    \end{bmatrix},
\end{align*}
then $T$ is a bistochastic matrix.
As explained in Lemma~\ref{lem:most-incompatible} for rank-1 projective measurements the incompatibility measure can be written as
\begin{align}
\label{eq:split-into-classes}
    \Upsilon_p(E,F) = \sum_{a,b=1}^d \widetilde{h}_p(t_{ab}) = \widetilde{h}_p(\tau) + \sum_{a =2}^d \widetilde{h}_p(t_{a1}) + \sum_{b =2}^d \widetilde{h}_p(t_{1b}) + \sum_{a,b =2}^d \widetilde{h}_p(t_{ab}).
\end{align}
Our goal is to derive bounds on this quantity which depend only on $\tau$ and $d$ and to do so we will use some properties of concave and Schur-concave functions introduced in Sec.~\ref{sec:concave_functions}.

To find an upper bound let us use the fact that each row and each column add up to 1, while all the entries add up to $d$. Applying Jensen's inequality independently to each of the three sums leads to
\begin{align}
\label{eq:ucp_upper_bound}
    \Upsilon_p(E,F) &\leq \widetilde{h}_p(\tau) + (d-1) \widetilde{h}_p\left(\frac{1}{d-1} \sum_{a =2}^d t_{a1}\right) + (d-1) \widetilde{h}_p\left(\frac{1}{d-1} \sum_{b =2}^d t_{1b}\right) + (d-1)^2 \widetilde{h}_p\left(\frac{1}{(d-1)^2} \sum_{a,b =2}^d t_{ab}\right) \nonumber\\
    &= \widetilde{h}_p(\tau) + 2 (d-1) \widetilde{h}_p\left(\frac{1 - \tau}{d-1}\right) + (d-1)^2 \widetilde{h}_p\left(\frac{d-2+\tau}{(d-1)^2}\right) \nonumber\\
    &= 2^\frac{1}{p} \left[ \sqrt{\tau(1-\tau)} + 2\sqrt{(1-\tau)(d-2+\tau)} + \sqrt{(d-2+\tau)(d^2 -3d+3-\tau)} \right].
\end{align}
It is easy to see that the resulting bound corresponds to simply maximising the right-hand side of Eq.~\eqref{eq:split-into-classes} over bistochastic matrices satisfying $t_{ab} \leq \tau$ for all $a, b$.

Note that the resulting upper bound is tight for the extreme values of $\tau$. For $\tau = \frac{1}{d}$ we obtain the value corresponding to $d$-dimensional MUBs, while for $\tau = 1$ we obtain the value corresponding to $(d - 1)$-dimensional MUBs. In both scenarios these upper bounds can be saturated by an appropriate choice of rank-1 projective measurements. For intermediate values of $\tau$ and $d \geq 3$ we should not expect the bound to be tight, because unistochastic matrices form a proper subset of bistochastic matrices. Finally, let us show that this upper bound is monotonically decreasing in $\tau$.

\begin{lem}
The upper bound given in Eq.~\eqref{eq:ucp_upper_bound} is a monotonically decreasing function of $\tau$ in the interval $[ \frac{1}{d}, 1 ]$.
\end{lem}
\begin{proof}
We will start by writing out the optimal bistochastic matrix  $T_{\textnormal{opt}}(\tau)$:
\begin{align*}
    T_{\textnormal{opt}}(\tau) := \begin{bmatrix}
    \tau & t_{r} & \cdots & t_{r}\\
    t_{r} & t_{s} & & \vdots\\
    \vdots & & \ddots & \\
    t_{r} & \cdots & & t_{s}
    \end{bmatrix},
\end{align*}
where $t_r(\tau) = \frac{1-\tau}{d-1}$ and $t_s(\tau) = \frac{d-2+\tau}{(d-1)^2}$. It is easy to check that for the extreme values we have $t_r(\frac{1}{d}) = t_s(\frac{1}{d}) = \frac{1}{d}$ and $t_r(1) = 0$, $t_s(1) = \frac{1}{d-1}$. To see that in the interval $[\frac{1}{d},1]$ we have $\tau \geq t_s \geq t_r$, it suffices to note that all three expressions are of the form $a \tau + b$ and since they coincide at $\tau = \frac{1}{d}$, it suffices to compare the coefficients of the linear term. Since $\widetilde{h}_p(t_{ab})$ is a strictly concave function of $t_{ab}$, the sum $\sum_{ab} \widetilde{h}_p(t_{ab})$ is a Schur-concave function (if we interpret it as a function from $\mathbb{R}^{d^{2}}$ to $\mathbb{R}$). Hence, to finish the proof we must show that the matrices $T_{\textnormal{opt}}(\tau)$ interpreted as vectors in $\mathbb{R}^{d^{2}}$ satisfy majorization in the sense that for $\tau \geq \tau'$ we have $T_{\textnormal{opt}}(\tau)\succ T_{\textnormal{opt}} (\tau')$. Since the ordering of coefficients is fixed (thanks to the relation $\tau \geq t_s \geq t_r$ valid for all relevant values of $\tau$), the inequalities specified in Eq.~\eqref{eq:majorization-condition} can be explicitly checked. Intuitively, the difference between the left-hand side and the right-hand side of Eq.~\eqref{eq:majorization-condition} starts off at $0$ for $k = 1$, then it increases until $k = 2d - 1$ (due to $t_s(\tau) \geq t_s(\tau')$) and then it decreases (due to $t_r(\tau) \leq t_r(\tau')$) but it only reaches 0 at $k = d^{2}$ (because in both cases the total sum of the entries equals $d$).
\end{proof}
Now let us come back to Eq.~\eqref{eq:split-into-classes} and use concavity to derive a lower bound. To find a lower bound on the sum $\sum_{a=2}^d \widetilde{h}_p(t_{a1})$ we minimize this expression over vectors in $\mathbb{R}^{d-1}$ which satisfy certain linear constraints. The allowed vectors form a convex set whose extremal points can be characterized using Lemma~\ref{lem:extremal_points} (simply set $s = 1 - \tau$ and $t = \tau$). Since the minimum must be achieved at one of these extremal points and all these points give exactly the same value, we conclude that:
\begin{align*}
    \sum_{a=2}^d \widetilde{h}_p(t_{a1}) &\geq \sum_{j=1}^{m_{r}} \widetilde{h}_p(\tau) + \widetilde{h}_p(1- \tau - m_{r} \tau) + \sum_{j=m_{r}+2}^{d-1} \widetilde{h}_p(0) = m_{r} \widetilde{h}_p(\tau) + \widetilde{h}_p(1- \tau - m_{r} \tau),
\end{align*}
where $m_{r} := \lfloor \frac{1-\tau}{\tau}\rfloor$ is the multiplicity (note the convention that if $m_{r} = 0$, then the first sum is empty and does not contribute) and we have used the fact that $\widetilde{h}_p(0) = 0$. The same argument can be applied to the second sum in Eq.~\eqref{eq:split-into-classes} to give exactly the same lower bound. Analogous argument applied to the last sum (use Lemma~\ref{lem:extremal_points} with $s = d - 2 + \tau$ and $t = \tau$) gives
\begin{align*}
%\sum_{b=2}^d \widetilde{h}_p(t_{1b}) &\geq m_{r} \widetilde{h}_p(\tau) + \widetilde{h}_p(1- \tau + m_{r} \tau), \\
    \sum_{a,b =2}^d \widetilde{h}_p(t_{ab}) &\geq \sum_{j=1}^{m_{s}} \widetilde{h}_p(\tau) + \widetilde{h}_p(d-2+\tau - m_{s} \tau) + \sum_{j=m_{s}+2}^{(d-1)^2} \widetilde{h}_p(0)\\
    &= m_{s} \widetilde{h}_p(\tau) + \widetilde{h}_p(d-2+\tau - m_{s} \tau).
\end{align*}
where $m_{s} := \lfloor \frac{d-2+\tau}{\tau} \rfloor$ (and we adopt analogous convention for the case of $m_{s} = 0$). Hence, the final expression for the lower bound reads:
\begin{align*}
    \Upsilon_p(E,F) \geq (1 + 2 m_{r} + m_{s}) \widetilde{h}_p(\tau) + 2 \widetilde{h}_p(1-\tau - m_{r} \tau) + \widetilde{h}_p(d-2+\tau - m_{s} \tau).
\end{align*}
Again, the resulting bound is tight for the extreme values. For $\tau = \frac{1}{d}$ we obtain the value corresponding to $d$-dimensional MUBs (i.e.~the upper and lower bounds match), while for $\tau = 1$ the lower bound vanishes, which corresponds to the case of compatible measurements. Note that the lower bound is not a smooth function, it exhibits kinks whose number increases with $d$. Nevertheless, from numerical evidence it seems to be an increasing function of $\tau$. In the special case of $d = 2$ the upper and lower bounds coincide for all values of $\tau$, because specifying $\tau$ essentially determines the overlap matrix. Having derived an upper and lower bound we can plot a region in the incompatibility--uncertainty space which is certainly forbidden. An example corresponding to $d = 3$ and incompatibility measure $\Upsilon_{1}$ can be found in Fig.~\ref{fig:uncertainty_principle}.
\begin{figure}[h!]
    \centering
    \includegraphics{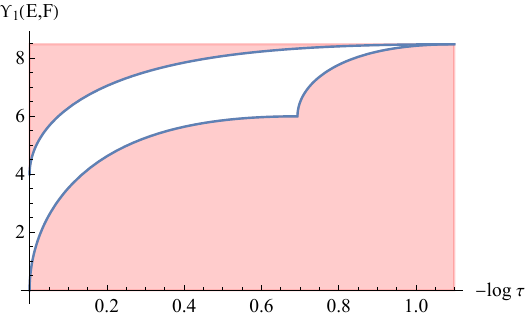}
    \caption{Region plot of the incompatibility--uncertainty space quantified by $\Upsilon_{1}$ and $- \log \tau$ for $d=3$ in the range $\tau \in [\frac{1}{3},1]$. The forbidden regions are marked in red.}
    \label{fig:uncertainty_principle}
\end{figure}
\section{Conclusions and open problems}
In this work we have introduced a new family of incompatibility measures based on non-commutativity. We showed that they are non-increasing under unitary operations and post-processing, but fail to do so under pre-processing. We have fully characterized the most incompatible pairs of measurements, which turn out to be a natural generalization of MUBs. We have also studied their behavior under different types of compositions.
To link our new measures with some existing results, we compared them with the generalized incompatibility robustness. We found that all pairs of measurements which are maximally incompatible according to non-commutativity measures are also maximally incompatible according to the generalized robustness measure. However, we do not know whether the reverse implication holds, as we do not have a full characterization of the maximally incompatible pairs for the robustness-based measure.
Finally, we investigated whether our incompatibility measures can be linked to QRAC performance and we found that a quantitative connection can be made for rank-1 projective measurements. More specifically, measurements exhibiting near-optimal QRAC performance must be close to achieving maximal incompatibility according to the non-commutativity-based measures. Similarly, we have shown that for rank-1 projective measurements our incompatibility measures can be related to entropic uncertainty relations. We have derived analytical upper and lower bounds on the incompatibility as a function of uncertainty. For the extremal cases our bounds are tight and we have given examples of measurements that saturate them. This allowed us to determine and visualise the forbidden regions in the incompatibility--uncertainty space.

Our work points at several open problems which we leave for future work. An important direction is to generalize our measures to more than two measurements (a feature that appears naturally in the robustness-based framework). Here, two potential approaches are to either define a symmetric function of pairwise incompatibilities or to find a suitable definition of a commutator for multiple operators. Another direction would be to extend the links to operational tasks, e.g.~could we prove a quantitative relation between non-commutativity-based measures and the generalized incompatibility robustness or could we extend the relation to QRAC performance for the most general class of quantum measurements? Finally, one could study the relation of our measures to other operational scenarios in which MUBs turn out to be relevant, e.g.~Bell nonlocality \cite{Kaniewski2019maximalnonlocality,Tavakolieabc3847} or the task of super-dense coding \cite{PhysRevA.88.062317,nayak2020rigidity}.

\section*{Acknowledgments}
We acknowledge fruitful discussions with Sébastien Designolle, Máté Farkas, Nicolás Gigena and Gabriel Pereira Alves.
The project ``Robust certification of quantum devices'' is carried out within the HOMING programme of the Foundation for Polish Science co-financed by the European Union under the European Regional Development Fund.

\providecommand{\noopsort}[1]{}\providecommand{\singleletter}[1]{#1}%
\end{document}